\documentclass[12pt,reqno,a4paper]{article}
\usepackage{amsmath,amssymb,dsfont,graphicx,latexsym,
epsf,cancel,epic,eepic,amsthm,tikz-cd,mathtools}
\usepackage[colorlinks=false,hidelinks]{hyperref}
\usepackage{mathpazo}
\usepackage[utf8]{inputenc}
\usepackage{tikz-cd}

\newtheorem{theorem}{Theorem}
\newtheorem{proposition}{Proposition}

\newcommand{\bbZ}{{\mathbb Z}}

\newcommand{\bbP}{{\mathbb P}}

\newcommand{\cD}{{\mathcal D}}

\newcommand{\cI}{{\mathcal I}}

\newcommand{\la}{{\langle}}
\newcommand{\ra}{{\rangle}}

\begin{document}

\title{Geometry of two- and three-dimensional integrable systems related to affine Weyl groups $W(E_8^{(1)})$ and $W(E_7^{(1)})$}
\author{Jaume Alonso \and Yuri B.\ Suris}
\date{February 23, 2026}
\maketitle

\begin{center}
{\small{
Institut f\"ur Mathematik, MA 7-1\\
Technische Universit\"at Berlin, Str. des 17. Juni 136,
10623 Berlin, Germany\\
E-mail: {\tt  jaume.alonsofernandez@tu-berlin.de, suris@math.tu-berlin.de}}}
\end{center}

\begin{abstract}
We find a general framework for the construction of birational involutions on two- and three-dimensional varieties obtained from $\mathbb P^2$, $\mathbb P^1\times \mathbb P^1$, and $\mathbb P^3$ by blow-up at nine, respectively eight points. Each such involution is based on a divisor class with a one-dimensional linear system with a generic element of genus zero. Classical Manin involutions represent the simplest particular case. Novel, more sophisticated cases identified here include birational involutions of $\mathbb P^2$ along conics and along nodal cubic curves, as well as birational involutions of $\mathbb P^3$ along quadratic cones and along Cayley nodal cubic surfaces. We prove a general formula for the induced action of geometric birational involutions on the respective Picard group, and give a general result about decomposition of translational elements of the respective affine Weyl group of symmetries into a product of two geometric birational involutions.
\end{abstract}

\section{Introduction}
\label{intro}
This paper continues our investigation of interrelations between constructions of classical algebraic geometry and discrete integrable systems, see \cite{PSWZ, W, ASW, ASW1, AS25}. The most well-known examples of geometric constructions we have in mind include Manin involutions for pencils of cubic curves \cite{Manin} and QRT maps based on pencils of biquadratic curves \cite{QRT1, QRT2, QRT book}. Discrete Painlev\'e equations were given an algebro-geometric interpretation and classification in the framework of the Sakai theory \cite{Sakai, KNY}, and in \cite{KMNOY, KMNOY1} they were related to constructions of the kind we speak about here. According to the Sakai theory, discrete Painlev\'e equations are birational maps of generalized Halphen surfaces which act on the Picard group as translational elements of the corresponding affine Weyl group of symmetries. For the most general such equation, the elliptic Painlev\'e equation, the symmetry group is $W(E_8^{(1)})$, and its translations are generated by $T_\alpha$ with the roots $\alpha=E_i-E_j$ and $\alpha=D-E_i-E_j-E_k$ (see section \ref{sect scenarios} for notations and further details). In \cite{KMNOY, KMNOY1}, a geometric construction was proposed for translations $T_{E_i-E_j}$, while finding a similar construction for $T_{D-E_i-E_j-E_k}$ remained open until our recent work \cite{AS25}, where it was settled with the help of some novel geometric involutions.

In the present work, empirical findings of \cite{AS25} generate an exhaustive conceptual explanation, which allows us to systematically derive novel geometric representations of two-dimensional integrable maps with the symmetry $W(E_8^{(1)})$ (autonomous versions of discrete Painlev\'e equations), as well as of three-dimensional integrable maps with the symmetry $W(E_7^{(1)})$ (autonomous versions of Takenawa's dynamical systems \cite{T}). The paper is organised as follows. In section \ref{sect scenarios}, we set up the scene by describing the three scenarios in which our theory works, namely special configurations of nine points in $\mathbb P^2$, of eight points in $\mathbb P^1\times\mathbb P^1$, and of eight points in $\mathbb P^3$. In section \ref{sec: involutions}, we define the divisor classes leading to geometric birational involutions, as well as give a detailed description of involutions themselves. In section \ref{sect action inv Picard} we prove an elegant formula for the induced action of those geometric involutions on the Picard group. In section \ref{sect translations}, we prove the general result about decomposition of translations from the affine Weyl group into products of reflections induced by geometric involutions. Some open questions and perspectives for the future work are listed in section \ref{sect conclusions}. Appendix \ref{App} contains some details of computations and formulas for the action of geometric involutions on the Picard group.

\section{Three scenarios}
\label{sect scenarios}

\paragraph{Scenario 1: Special configuration of nine points in $\mathbb P^2$.}

Here, $P_1,\ldots,P_9$ are nine points in $\mathbb P^2$, which are in special position in the sense that they support a pencil of cubic curves, but generic otherwise (no thee points collinear, no six points on a conic). We denote
\begin{equation}
X=Bl_{P_1,\ldots,P_9}(\mathbb P^2).
\end{equation}
The Picard group of $X$ is
\begin{equation}\label{Pic sc 1}
	\text{Pic}(X) = \bbZ D \oplus \bbZ E_1 \oplus \ldots \oplus \bbZ E_9,
\end{equation} 
where $D$ is the class of the total transform of a generic line, and $E_i$ is the class of the exceptional divisor of the blow-up at the point $P_i$. 
The intersection number defines a non-degenerate scalar product on $\text{Pic}(X)$, with the only non-vanishing products among the basis elements being
\begin{equation}
   \langle D ,D \rangle = 1,\qquad \langle E_i,E_i \rangle = -1, \quad i=1, \ldots,9.
\end{equation} 
We set 
\begin{equation}
	\delta =  -K_X=3D - E_1 - \ldots - E_9
\end{equation}
as the anti-canonical divisor of $X$, divisors of this class being the proper images of cubic curves of the pencil through $P_1,\ldots,P_9$ in $\mathbb P^2$. We have $\langle\delta,\delta\rangle=0$. We set
\begin{equation}\label{Q}
Q=\{\alpha\in \text{Pic}(X): \langle \alpha,\delta\rangle =0\},
\end{equation}
which one can directly check that it is a root lattice of the type $E_8^{(1)}$. The formula
\begin{equation}\label{R}
R=\{\alpha\in \text{Pic}(X): \langle \alpha,\delta\rangle =0, \;  \langle \alpha,\alpha\rangle =-2 \}.
\end{equation}
defines a root system of the affine type $E_8^{(1)}$. Any element of $R$ is equivalent $\pmod {\bbZ \delta}$ to one of the roots given by
\begin{equation}
	E_i - E_j \; (i < j) \quad \text{and} \quad D - E_i - E_j - E_k \; (i < j < k ),
\end{equation} 
or to one opposite of those.

\paragraph{Scenario 2: Special configuration of eight points in $\mathbb P^1\times \mathbb P^1$.}

Here, $P_1,\ldots,P_8$ are eight points in $\mathbb P^1\times \mathbb P^1$, which are in special position in the sense that they support a pencil of biquadratic curves, but generic otherwise. We set
\begin{equation}\label{Y}
X=Bl_{P_1,\ldots,P_8}(\bbP^1\times\bbP^1),
\end{equation}
with the Picard group 
\begin{equation}\label{Pic Y}
{\rm Pic}(X)=\bbZ H_1\oplus\bbZ H_2\oplus\bbZ E_1\oplus\ldots\oplus\bbZ E_8.
\end{equation}
Here $H_1$, $H_2$  are the divisor classes of the total transforms of a generic vertical, resp. horizontal line, while $E_i$ are exceptional divisors of the blow-up of $\bbP^1\times\bbP^1$ at $P_i$. The Minkowski scalar product (intersection number) has the only non-vanishing products of the basis elements 
\begin{equation}\label{scal prod X}
\langle H_1,H_2\rangle=1,\quad \langle E_i,E_i\rangle =-1, \;\; i=1,\ldots 8.
\end{equation}
This lattice is isomorphic (and isometric) to the previous one. We set
\begin{equation}\label{delta Y}
\delta=-K_X=2H_1+2H_2-E_1-\ldots-E_8,
\end{equation}
and define $Q$ and $R$ by the same formulas \eqref{Q} and \eqref{R} as before (these are still a root lattice and a root system of the type $E_8^{(1)}$). Some typical elements  of $R$ are:
\begin{equation}
H_1-H_2,\quad E_i-E_j, \quad H_1-E_i-E_j,\quad H_2-E_i-E_j, \quad H_1+H_2-E_i-E_j-E_k-E_\ell. \quad
\end{equation}

\paragraph{Scenario 3: Special configuration of eight points in $\mathbb P^3$.}
Here, $P_1,\ldots,P_8$ are eight points in $\mathbb P^3$, which are in special position in the sense that they support a two-dimensional family (a net) of quadrics, but generic otherwise (no four points coplanar, etc.). We denote
\begin{equation}
X=Bl_{P_1,\ldots,P_8}(\mathbb P^3).
\end{equation}
The Picard group of $X$ is
\begin{equation}
	\text{Pic}(X) = \bbZ \Pi \oplus \bbZ E_1 \oplus \ldots \oplus \bbZ E_8,
\end{equation} 
where $\Pi$ is the class of the total transform of a generic plane, and $E_i$ is the class of the exceptional divisor of the blow-up at the point $P_i$. 
We set
\begin{equation}
	\delta = -\frac{1}{2}K_X=2 \Pi - E_1 - \ldots - E_8
\end{equation}
(divisors of this class being the proper images of quadrics of the net through $P_1,\ldots,P_8$ in $\mathbb P^3$). 
We use the scalar product on the Picard group defined in \cite{T}, namely
\begin{equation*}
	\langle \cD,\cD' \rangle := \cD \cdot \cD' \cdot \delta,
\end{equation*} 
where for any $\cD,\cD',\cD'' \in \text{Pic}(X)$, $\cD \cdot \cD' \cdot \cD''$ denotes the intersection number $\text{Pic}(X) \times \text{Pic}(X) \times \text{Pic}(X) \to \bbZ$. The only non-zero products among the basis elements are
\begin{equation}
   \langle \Pi ,\Pi \rangle = 2,\quad \langle E_i,E_i \rangle = -1, \; i=1, \ldots,8.
\end{equation} 
With this scalar product, we have $\langle \delta,\delta\rangle=0$, and formulas \eqref{Q}, \eqref{R} define a root lattice, resp. a root system of the affine type $E_7^{(1)}$.  Any element of $R$ is equivalent$\pmod {\bbZ \delta}$ to one of the roots given by
\begin{equation*}
	E_i - E_j \; (i < j) \quad \text{and} \quad \Pi - E_i - E_j - E_k-E_\ell \quad (i < j < k < \ell),
\end{equation*} 
or to one opposite of those. 

\section{General construction of geometric birational involutions}
\label{sec: involutions}

We are interested in birational automorphisms of $X$ which induce the action of \emph{translational elements} of the Weyl group on $\text{Pic}(X)$.  Our main finding is a unified construction of geometric involutions, which is valid in all three scenarios, such that the translations from the Weyl group are induced by compositions of two involutions. The geometric involutions are related to certain divisor classes with one-dimensional linear systems. More precisely, we define
\begin{equation}
B=\{\beta\in\text{Pic}(X): \langle \beta,\beta\rangle=0,\; \langle \beta,\delta\rangle =2\}.
\end{equation}
It is well known (see, e.g., \cite[formula (3.33)]{KNY}) that
\begin{equation}
2d(\beta)=\langle\beta,\beta\rangle + \langle\delta,\beta\rangle , \quad
2g(\beta)-2=\langle\beta,\beta\rangle - \langle\delta,\beta\rangle,
\end{equation}
where $d(\beta)$ is the dimension of the linear system of the divisor class $\beta$, and $g(\beta)$ is the genus of a generic divisor of this class. Thus, relations $\langle \beta,\beta\rangle=0$ and $\langle \beta,\delta\rangle=2$ are equivalent to
\begin{equation}
d(\beta)=1, \quad g(\beta)=0,
\end{equation}
The relation $d(\beta)=1$ means that for any generic point $P$ (not one of the base points), there is exactly one divisor $\beta_P$ of the class $\beta$ through this point. 
\smallskip

{\bf Scenarios 1, 2:} There is exactly one divisor $\delta_P$ of the class $\delta$ through $P$. The intersection number of $\delta_P$ with $\beta_P$ is $\langle\delta,\beta\rangle=2$, so there is exactly one further intersection point $\widetilde P$ (different from $P$). 
\smallskip

{\bf Scenario 3:} There is a pencil (one-parameter) family of divisors of the class $\delta$ through $P$, spanned by $\delta_P^{(1)}$ and $\delta_P^{(2)}$, say. The triple intersection number of $\delta_P^{(1)}$, $\delta_P^{(2)}$ and $\beta_P$ is
$$
\delta\cdot\delta\cdot\beta=\langle \delta,\beta\rangle=2,
$$
so that again there is there is exactly one further intersection point $\widetilde P$ (different from $P$). 
\smallskip

In both cases we define the map $\cI_\beta$ by declaring
\begin{equation}
\cI_\beta(P)=\widetilde P.
\end{equation} 

There follows the list of concrete realizations of this general construction. For this, we present the computations of the intersection numbers in the non-blown-up spaces $\mathbb P^2$, $\mathbb P^1\times\mathbb P^1$, and $\mathbb P^3$, respectively. In each of the scenarios, we indicate some typical divisor classes from the set $B$. Some of the examples (those for which we do not provide references) seem to be new.
\medskip

{\bf Scenario 1.} For any $P\in\mathbb P^2\setminus\{P_1,\ldots,P_9\}$, we denote by $C_P$ the unique cubic curve of the pencil passing through $P$ (with the proper transform $\delta_P$), and by $L_P$ the unique curve of the class $\beta$ through $P$ (with the proper transform $\beta_P$).

\begin{itemize}
\item[a)] $\beta=D-E_i$, the class of lines through $P_i$. Here, $L_P$ is the line $(PP_i)$. It intersects $C_P$ at three points, two of them being  $P_i$ and $P$. The third one is $\widetilde P=:\cI_{\beta}(P)$. This is the classical {\em  Manin involution} \cite{Manin}.

\item[b)] $\beta=2D-E_i-E_j-E_k-E_\ell$, the class of conics through $P_i$, $P_j$, $P_k$, $P_\ell$. Thus, $L_P$ is the conic through the five points $P_i$, $P_j$, $P_k$, $P_\ell$ and $P$. It intersects $C_P$ at $2\cdot 3=6$ points, five of them being the ones just listed. The sixth one is $\widetilde P=:\cI_{\beta}(P)$. This is the involution along conics introduced in \cite{AS25}.

\item[c)] $\beta=3D-2E_i-E_j-E_k-E_\ell-E_m-E_n$, the class of cubic curves with a double point at $P_i$ and passing through further 5 points $P_j$, $P_k$, $P_\ell$, $P_m$, $P_n$. Thus, $L_P$ is the unique such curve through $P$. It intersects $C_P$ at $3\cdot 3=9$ points, eight of them being the ones just listed (with $P_i$ counted twice). The 9-th one is $\widetilde P=:\cI_{\beta}(P)$. 

\end{itemize}

{\bf Scenario 2.}  For any $P\in\mathbb P^1\times\mathbb P^1\setminus\{P_1,\ldots,P_8\}$, we denote by $C_P$ the unique biquadratic curve of the pencil passing through $P$ (with the proper transform $\delta_P$).
\begin{itemize}
\item[a)] $\beta=H_1$, the class of vertical lines. In particular, $L_P$ is the vertical line through $P$. It intersects $C_P$ at two points, one of them being $P$. The second one is $\widetilde P=:\cI_{\beta}(P)$. This is the \emph{vertical QRT switch}. The horizontal QRT switch is defined similarly, with $\beta=H_2$, the class of horizontal lines \cite{QRT1, QRT2, QRT book}.

\item[b)] $\beta=H_1+H_2-E_i-E_j$, the class of (1,1)-curves through $P_i$ and $P_j$. The unique such curve through $P$ is denoted $L_P$. It intersects $C_P$ at $1\cdot 2+1\cdot 2=4$ points, three of them being $P_i$, $P_j$ and $P$. The fourth one is $\widetilde P=:\cI_{\beta}(P)$. This the involution along (1,1)-curves introduced in \cite{AS25}.

\item[c)] $\beta=2H_1+H_2-E_i-E_j-E_k-E_\ell$, the class of (2,1)-curves through $P_i$, $P_j$, $P_k$, $P_\ell$. The unique such curve through $P$ is denoted $L_P$. It intersects $C_P$ at $2\cdot 2+1\cdot 2=6$ points, five of them being $P_i$, $P_j$, $P_k$, $P_\ell$ and $P$. The sixth one is $\widetilde P=:\cI_{\beta}(P)$. This the involution along (2,1)-curves introduced in \cite{AS25}. The involution along (1,2)-curves is defined similarly.

\item[d)] $\beta=2(H_1+H_2)-E_i-E_j-E_k-E_\ell-2E_m$, the class of (2,2)-curves passing through $P_i$, $P_j$, $P_k$, $P_\ell$ and having a double point at $P_m$. The unique such a curve through $P$ is denoted $L_P$. It intersects $C_P$ at $2\cdot 2+2\cdot 2=8$ points, seven of them being $P_i$, $P_j$, $P_k$, $P_\ell$, $P_m^2$ and $P$. The 8-th one is $\widetilde P=:\cI_{\beta}(P)$.

\item[e)] $\beta=3(H_1+H_2)-2E_i-2E_j-2E_k-2E_\ell-E_m-E_n$, the class of (3,3)-curves with double points at $P_i$, $P_j$, $P_k$, $P_\ell$ and simple points at $P_m$ and $P_n$. The unique such a curve through $P$ is denoted $L_P$. It intersects $C_P$ at $3\cdot 2+3\cdot 2=12$ points, eleven of them being $P_i^2$, $P_j^2$, $P_k^2$, $P_\ell^2$, $P_m$, $P_n$ and $P$. The 12-th one is $\widetilde P=:\cI_{\beta}(P)$.
\end{itemize}

{\bf Scenario 3.} For any $P\in\mathbb P^3\setminus\{P_1,\ldots,P_8\}$, the quadrics of the net through the point $P$ form a pencil, whose base set is a space curve of degree 4 (the intersection of any two quadrics of the pencil). We denote this base curve by $C_P$. We denote by $L_P$ the unique surface of the class $\beta$ through $P$.
\begin{itemize}

\item[a)] $\beta=\Pi-E_i-E_j$, the class of planes through $P_i$ and $P_j$. Here, $L_P$ is the unique plane through $P_i$, $P_j$ and $P$. It intersects $C_P$ at four points, three of them being $P_i$, $P_j$ and $P$. The fourth one is $\widetilde P=:\cI_{\beta}(P)$. This involution was described in the concluding remarks of \cite{ASW1}.

\item[b)] $\beta=2\Pi-E_i-E_j-E_k-E_\ell-2E_m$, the class of quadratic cones with the tip at $P_m$, passing through $P_i$, $P_j$, $P_k$, $P_\ell$. The unique such a cone through $P$ is denoted by $L_P$. It intersects $C_P$ at $2\cdot 4=8$ points, seven of them being $P_i$, $P_j$, $P_k$, $P_\ell$, $P_m^2$ and $P$. The 8-th one is $\widetilde P=:\cI_{\beta}(P)$.

\item[c)] $\beta=3\Pi-2E_i-2E_j-2E_k-2E_\ell-E_m-E_n$, the class of Cayley nodal cubic surfaces with double points at $P_i$, $P_j$, $P_k$, $P_\ell$ and simple points at $P_m$ and $P_n$. The unique such nodal Cayley cubic through $P$ is denoted $L_P$. It intersects $C_P$ at $3\cdot 4=12$ points, eleven of them being $P_i^2$, $P_j^2$, $P_k^2$, $P_\ell^2$, $P_m$, $P_n$ and $P$. The 12-th one is $\widetilde P=:\cI_{\beta}(P)$.
\end{itemize}

\section{Action of geometric birational involutions on the Picard group}
\label{sect action inv Picard}

\begin{theorem}\label{th I beta on Pic}
For $\beta\in B$, the birational involution $\cI_\beta$ induce the following linear involution on the Picard group $I_\beta=(\cI_\beta)_*: {\rm Pic}(X)\to{\rm Pic}(X)$ :
\begin{equation}\label{I beta on Pic}
I_\beta(\lambda)=-\lambda+\langle \lambda,\beta\rangle \delta + \langle \lambda,\delta\rangle \beta.
\end{equation}
\end{theorem}

\subsection{Scheme of the proof for scenarios 1 and 2}

In this section, we extend the map $I_\beta$ by linearity from the Picard group \eqref{Pic sc 1} to the vector space
\begin{equation}\label{V}
V = \mathbb Q D \oplus \mathbb Q E_1 \oplus \ldots \oplus \mathbb Q E_9
\end{equation}
(and similarly in the scenarios 2 and 3). Since the right-hand side of \eqref{I beta on Pic} is linear with respect to $\lambda$, it suffices to verify the formula for the elements of some basis of $V$ (ten linearly independent elements). For two elements we have a general result:
\begin{equation}\label{I beta on delta}
I_\beta(\delta)=\delta, \quad I_\beta(\beta)=\beta, 
\end{equation}
since, by definition, all divisors of the classes $\delta$ and $\beta$ are invariant under $\cI_\beta$. This coincides with \eqref{I beta on Pic} for $\delta$ and $\beta$ in view of $\langle\beta,\delta\rangle=2$, $\langle\delta,\delta\rangle=0$, and $\langle\beta,\beta\rangle=0$. The two-dimensional space ${\rm span}(\beta,\delta)$ consists of fixed points of $I_\beta$.

Further, for any $j$ such that $\langle \beta, E_j\rangle =0$, we have:
\begin{equation}\label{I beta Ej}
I_\beta(E_j)=\beta-E_j.
\end{equation}
Indeed, by definition, $\cI_\beta$ blows down the divisor of the class $\beta$ passing through $P_j$ to the point $P_j$, and, since it is an involution, it blows up the point $P_j$ to the divisor of the class $\beta$ passing through $P_j$. We express this by saying that these $E_j$ and $\beta-E_j$ belong to
\begin{equation}\label{V1}
V_1=\big\{\lambda\in V: I_\beta(\lambda)=-\lambda+\beta\big\}.
\end{equation}
Formula \eqref{I beta Ej} coincides with \eqref{I beta on Pic} for $E_j$ and for $\beta-E_j$, since $\langle \beta,E_j\rangle=0$ and $\langle \delta,E_j\rangle=1$.

It turns out that in the scenarios 1 and 2 one can identify further elements of $V_1$. We will show, on the case-by-case basis (see details in Appendix \ref{App}), that there exist splittings
\begin{equation}\label{beta split}
\beta=\gamma_1+\gamma_2
\end{equation}
where $\gamma_1$, $\gamma_2$ are two $(-1)$-curves, so that $\langle\gamma_1,\gamma_1\rangle= \langle\gamma_2,\gamma_2\rangle=-1$, satisfying 
\begin{equation}\label{gamma prop}
\langle \beta,\gamma_1\rangle=\langle \beta,\gamma_2\rangle=0, \quad \langle \delta,\gamma_1\rangle=\langle \delta,\gamma_2\rangle=1. 
\end{equation}
For such a splitting, we derive: 
\begin{equation}\label{trick 1}
I_\beta(\gamma_1)=\gamma_2, \quad I_\beta(\gamma_2)=\gamma_1.
\end{equation}
Indeed, by definition, for any curve of the class $\beta$, the image $\widetilde{P}=I_\beta(P)$ belongs to the unique curve $L_P$ of this class through $P$. We slightly abuse notation by writing $\gamma_1$, $\gamma_2$ for the unique curves of the corresponding classes. If $P\in \gamma_1$ then $L_P=\gamma_1\cup\gamma_2$, and then it follows that generically $\widetilde{P}\in\gamma_2$, and vice versa. Now the key point is that equations \eqref{trick 1} can be put as
\begin{equation}\label{trick 2}
I_\beta(\gamma_1)=\beta-\gamma_1, \quad I_\beta(\gamma_2)=\beta-\gamma_2,
\end{equation}
or $\gamma_1,\gamma_2\in V_1$. Because of \eqref{gamma prop}, we see that \eqref{I beta on Pic} holds true for for $\gamma_1$, $\gamma_2$. 

Thus, the affine subspace $V_1$ in $V$ contains $E_j$, $\beta-E_j$ with $\langle \beta,E_j\rangle=0$ and $\gamma_1$, $\gamma_2$ from the splittings \eqref{beta split}, and for all these elements equation \eqref{I beta on Pic} is satisfied. To finish the proof, it suffices to show that these elements span a space of dimension $\ge 8$.  By a simple case-by-case check we see that the elements $E_j$ with $\langle \beta,E_j\rangle=0$ can be always complemented by suitable elements $\gamma_1$ to eight linearly independent ones, see Appendix \ref{App} for details. \qed 
\medskip

{\bf Remark.} Formula \eqref{I beta on Pic} implies that the affine space $V_1$ admits an alternative description as
\begin{equation}\label{V1 alt}
V_1=\big\{\lambda\in V: \langle \beta,\lambda\rangle=0, \; \langle \delta,\lambda\rangle=1\big\}.
\end{equation}
Differences of elements of $V_1$ form the vector space
\begin{equation}\label{eigenspace -1}
V_0=\big\{\lambda\in V:  I_\beta(\lambda)=-\lambda\big\}
\end{equation}
(the eigenspace of $I_\beta$ for the eigenvalue $-1$), and, according to formula \eqref{I beta on Pic}, this space admits an alternative description
\begin{equation}\label{V0}
V_0=\big\{\lambda\in V: \langle \beta,\lambda\rangle=\langle \delta,\lambda\rangle=0\big\}=\big({\rm span}(\beta,\delta)\big)^\perp.
\end{equation}
Therefore, $\dim V_0=8$. One can find a basis of $V_0$ consisting of eight roots from $R$, see Appendix \ref{App}. 

\subsection{Proof for scenario 3}

We obtain the result for the scenario 3 without further computations, based on the following observation.
\begin{proposition}\label{prop 2D to 3D}
For any $\beta\in B$ in the scenario 3, the corresponding involution $\cI_\beta$ induces on every non-degenerate quadric $Q$ of the net through $P_1,\ldots,P_8$ the involution $\cI_{\widetilde\beta}$ of the scenario 2, where $Q$ is interpreted as $\mathbb P^1\times\mathbb P^1$ via the Segre embedding, and the element $\widetilde\beta$ is formally obtained from $\beta$ by replacing $\Pi$ through $H_1+H_2$. 
\end{proposition}
\begin{proof}
Each non-degenerate quadric $Q$ of the net through $P_1,\ldots,P_8$ in the scenario 3 can be interpreted (via the Segre embedding) as $\mathbb P^1\times \mathbb P^1$ with eight points $P_1,\ldots,P_8$ which support a pencil of biquadratic curves (intersections of $Q$ with other members of the net). Exceptional divisors $E_1,\ldots, E_8$ of $X=Bl_{P_1,\ldots,P_8}(\mathbb P^3)$ are restricted to the exceptional divisors (denoted by the same letters) of $Bl_{P_1,\ldots,P_8}(Q)\simeq Bl_{P_1,\ldots,P_8}(\mathbb P^1\times \mathbb P^1)$. Intersections of $Q$ with the surfaces of the class $\beta$ in the scenario 3 are nothing but curves of the corresponding class $\widetilde\beta$ in the scenario 2. Indeed, each surface $S$ of the class 
$$
\beta=m\Pi-\sum m_iE_i
$$ 
is of degree $m$. It intersects $Q$ along a space curve $L$ of degree $2m$. This curve intersects each generator of $Q$ at exactly $m$ points, therefore, in the interpretation of $Q$ as $\mathbb P^1\times\mathbb P^1$  the curve $L$ is a $(m,m)$-curve of the class
$$
\widetilde\beta=m(H_1+H_2)-\sum m_iE_i.
$$ 
Finally, we notice that the sublattice 
$$
 \bbZ (H_1+H_2) \oplus \bbZ E_1 \oplus \ldots \oplus \bbZ E_8
$$
of the Picard group in the scenario 2 is isometric to the Picard group in the scenario 3, due to 
$$
\langle H_1+H_2, H_1+H_2\rangle = 2=\langle \Pi,\Pi\rangle
$$
(with the different meanings of the scalar product symbol on the both sides).
\end{proof}

\section{Geometric interpretation of translations in the affine Weyl group}
\label{sect translations}

Recall the Kac formula \cite{Kac} for the action of the translation $T_\alpha$ on $\text{Pic}(X)$, corresponding to $\alpha\in Q$:
\begin{equation}\label{Kac}
	T_\alpha (\lambda) = \lambda + \langle \delta, \lambda \rangle \alpha - \left( \dfrac{1}{2} \la\delta, \lambda \ra \la \alpha, \alpha \ra + \la \alpha, \lambda \ra\right) \delta.
\end{equation}

\begin{theorem}
For any $\beta_1,\beta_2\in B$, we have:
\begin{equation}
I_{\beta_1}\circ I_{\beta_2}=T_{\beta_1-\beta_2}.
\end{equation} 
Thus, the birational map $\cI_{\beta_1}\cdot\cI_{\beta_2}$ induces the translation $T_{\beta_1-\beta_2}$ on the Picard group.
\end{theorem}
\begin{proof}
We compute: 
$$
I_{\beta_1}(\delta)=\delta, \quad  I_{\beta_1}(\beta_2)=-\beta_2+\langle \beta_1,\beta_2\rangle \delta+2\beta_1, 
$$ 
therefore
\begin{eqnarray*}
I_{\beta_1}\circ I_{\beta_2}(\lambda) & = & I_{\beta_1}\big(-\lambda+\langle \lambda,\beta_2\rangle \delta + \langle \lambda,\delta\rangle \beta_2\big)\\
& = & \lambda-\langle \lambda,\beta_1\rangle \delta - \langle \lambda,\delta\rangle \beta_1\\
&   & +\langle \lambda,\beta_2\rangle \delta +\langle \lambda,\delta\rangle (-\beta_2+\langle \beta_1,\beta_2\rangle \delta+2\beta_1)\\
& = & \lambda +\langle \lambda,\delta\rangle (\beta_1-\beta_2)-\Big(\langle\lambda, \beta_1-\beta_2 \rangle -\langle \lambda,\delta\rangle\langle\beta_1,\beta_2\rangle \Big)\delta\\
& = & \lambda +\langle \lambda,\delta\rangle (\beta_1-\beta_2)-\Big(\langle\lambda, \beta_1-\beta_2 \rangle +\frac{1}{2}\langle \lambda,\delta\rangle\langle\beta_1-\beta_2,\beta_1-\beta_2\rangle \Big)\delta\\
& = & T_{\beta_1-\beta_2}(\lambda).
\end{eqnarray*}
This finishes the proof.
\end{proof}

We are mostly interested in the geometric representation $\cI_{\beta_1}\circ\cI_{\beta_2}$ of translations $T_\alpha$ with $\alpha\in R$. Below, we provide examples of splittings $\alpha=\beta_1-\beta_2$ of the roots $\alpha\in R$ (they are by no means unique).
\medskip

{\bf Scenario 1.}
\begin{center}
\begin{tabular}{|c|c|c|}
\hline
$\alpha$ & $\beta_1$ & $\beta_2$ \\
\hline\hline
$E_i-E_j$ & $D-E_j$ & $D-E_i$\\
\hline
$D-E_i-E_j-E_k$ & $2D-E_i-E_j-E_k-E_\ell$ & $D-E_\ell$\\
\hline
$D-E_i-E_j-E_k$ & $3D-2E_i-E_j-E_k-E_\ell-E_m-E_n$ & $2D-E_i-E_\ell-E_m-E_n$\\
\hline
\end{tabular}
\end{center}
\medskip

{\bf Scenario 2.}
\begin{center}
\begin{tabular}{|c|c|c|}
\hline
$\alpha$ & $\beta_1$ & $\beta_2$ \\
\hline\hline
$H_1-H_2$ & $H_1$ & $H_2$\\
\hline
$H_1-E_i-E_j$ & $H_1+H_2-E_i-E_j$ & $H_2$\\
\hline
$H_2-E_i-E_j$ & $H_1+H_2-E_i-E_j$ & $H_1$\\
\hline
$H_1+H_2-E_i-E_j-E_k-E_\ell$ & $2H_1+H_2-E_i-E_j-E_k-E_\ell$ & $H_1$\\
\hline
$H_1+H_2-E_i-E_j-E_k-E_\ell$ & $H_1+2H_2-E_i-E_j-E_k-E_\ell$ & $H_2$\\
\hline
$E_i-E_j$ & $H_1+H_2-E_j-E_k$ & $H_1+H_2-E_i-E_k$\\
\hline
$H_1+H_2-E_i-E_j-E_k-E_\ell$ & $2(H_1+H_2)-2E_i-E_j-E_k $ & $H_1+H_2-E_i-E_m$\\ & $-E_\ell-E_m$ & \\
\hline
$H_1+H_2-E_i-E_j-E_k-E_\ell$ & $3(H_1+H_2)-2E_i-2E_j$ & $2(H_1+H_2)-E_i-E_j$ \\
 & $-2E_k-2E_m-E_\ell-E_n$ & $-E_k-E_n-2E_m$ \\
\hline
\end{tabular}
\end{center}
\medskip

{\bf Scenario 3.}
\begin{center}
\begin{tabular}{|c|c|c|}
\hline
$\alpha$ & $\beta_1$ & $\beta_2$ \\
\hline\hline
$E_i-E_j$ & $\Pi-E_j-E_k$ & $\Pi-E_i-E_k$\\
\hline
$\Pi-E_i-E_j-E_k-E_\ell$ & $2\Pi-2E_i-E_j-E_k-E_\ell-E_m$ & $\Pi-E_i-E_m$\\
\hline
$\Pi-E_i-E_j-E_k-E_\ell$ & $3\Pi-2E_i-2E_j-2E_k-2E_m$ & $2\Pi-E_i-E_j-E_k-E_n-2E_m$ \\ & $-E_\ell-E_n$ & \\
\hline
\end{tabular}
\end{center}
\medskip

Thus, we have a systematic geometric construction for birational maps inducing translations $T_\alpha$ for roots $\alpha$ of the respective affine root systems $E_8^{(1)}$ and $E_7^{(1)}$.

\section{Conclusions}
\label{sect conclusions}

In this paper, we work out the general theory behind the findings of \cite{AS25}, and extend it to the three-dimensional integrable systems. In particular, we find geometric interpretation for integrable birational 3D maps related to special configurations of eight points in $\mathbb P^3$ \cite{T} (eight points supporting a net of quadrics). As we show, one can lift to 3D all translational maps of the corresponding QRT hierarchy corresponding to divisor classes $\beta\in B$ with $\langle \beta, H_1-H_2\rangle =0$. The latter restriction is essential. As we pointed out in \cite{ASW, ASW1, AS2}, the natural lifts to 3D of the standard QRT involutions corresponding to the divisor classes $\beta=H_1$, $\beta=H_2$ are not birational maps of $\mathbb P^3$. Rather, they branch over degenerate quadrics (cones) of the underlying net. This branching behaviour is quite natural: the corresponding constructions impose the assignment of generators of the quadrics to one of the classes $H_1$, $H_2$, and such an assignment breaks down (and branches) over the cones. This correlates with the fact that there are ``less'' birational involutions and translations in dimension 3 than in dimension 2 (the symmetry group $W(E_7^{(1)})$ is ``smaller'' than $W(E_8^{(1)})$). Elaborating on 3D analogs of geometric involutions with $\langle \beta, H_1-H_2\rangle \neq 0$ is left for the future work, the first step being the construction of 3D QRT involutions in \cite{ASW}. 
\smallskip

Some further natural directions of further investigations are:
\begin{itemize}
\item Extend constructions of the present work to the ``non-autonomous'' case, i.e., of generic point configurations (nine points in $\mathbb P^2$ supporting a unique cubic curve in the scenario 1, eight points in $\mathbb P^1\times\mathbb P^1$ supporting a unique biquadratic curve in the scenario 2, and eight points in $\mathbb P^3$ supporting a pencil of quadrics in the scenario 3). The first two scenarios correspond to elliptic Painlev\'e equations, and the translations for $\alpha=E_i-E_j$ were given a geometric interpretation in \cite{KMNOY, KMNOY1}. Geometric interpretation for other roots remains open. Likewise, the geometry of the non-autonomous maps from \cite{T} is not yet understood.

\item Specialize constructions of the present work to even more special point configurations, leading to discrete Painlev\'e equations with smaller symmetry groups and to their 3D counterparts (the latter remain largely unknown, compare \cite{T}). 

\item Combine constructions of the present work with {\em Painlev\'e deformations maps} from \cite{ASW1, AS2}, to achieve an alternative understanding of discrete Painlev\'e equations.

\item Integrability of our maps refers also to further geometric attributes, like invariant measures. This aspect has been dealt with in \cite{ASvolume} for 3D QRT maps, and this investigation should be extended to the novel maps introduced in the present work.

\end{itemize}
All this will be a subject of further work.



\begin{appendix}

\section{Formulas for the action of geometric involutions on the Picard group}
\label{App}

\subsection{Scenario 1}

\subsubsection{\boldmath$\beta=D-E_i$} 

Here, formula \eqref{I beta on Pic} has been verified for $\beta$, $\delta$ and for $E_j$ with $j\neq i$ (eight in total). These ten elements are a basis of $V$, which finishes the proof of \eqref{I beta on Pic} in the present case. This allows us to compute the matrix of $I_\beta$ in the standard basis of ${\rm Pic}(X)$:
\begin{alignat}{2}
I_{\beta}(D)  & \; =-D+3\beta+\delta & & = 5D-4E_i-\sum_{k\neq i} E_k,\label{Ii 0}\\
I_{\beta}(E_i)  & \; =-E_i+\beta+\delta&& =4D-3E_i-\sum_{k\neq i}E_k, \label{Ii i}\\
I_{\beta}(E_j)  & \;=-E_j+\beta && =D -E_i-E_j, \quad j\neq i. \label{Ii j}
\end{alignat}
This action of $I_\beta$ on ${\rm Pic}(X)$ was computed in \cite{Manin}.

\subsubsection{\boldmath$\beta=2D-E_1-E_2-E_3-E_4$}

We fix the indices for a better readability since there is, of course, no loss of generality. In this case, the space $V_1$ contains $E_j$ with $j=5,\ldots,9$ (five in total). Further elements of $V_1$ are given by the following decompositions of the kind \eqref{beta split}:
\begin{equation}\label{I1234 split}
\gamma_1=D-E_i-E_k, \quad \gamma_2=D-E_j-E_\ell,  \quad \{i,j,k,\ell\}=\{1,2,3,4\},
\end{equation}
and elements $\gamma_1$ with $(i,k)=(1,2), (1,3), (1,4)$ complement $E_j$ to eight linearly independent ones. This proves Theorem \ref{th I beta on Pic} in the present case. 
Vector space $V_0$ contains the following roots:
\begin{eqnarray}
\lambda & = & E_i-E_j, \quad i,j\in\{1,2,3,4\}, \label{I1234 V0 1}\\
\lambda & = & D-E_i-E_j-E_k, \quad i,k\in\{1,2,3,4\}, \;\; j\in\{5,\ldots,9\}. \label{I1234 V0 2}
\end{eqnarray}
Eight linearly independent ones are given, e.g., by \eqref{I1234 V0 1} with $i=1$ and $j=2,3,4$ and by \eqref{I1234 V0 2} with $i=1$, $k=2$ and $j=5,\ldots,9$. 

We compute $I_\beta$ in the standard basis of ${\rm Pic}(X)$:
\begin{alignat}{2}
I_\beta(D) & \;= -D+ 3\beta+2\delta &&\;=11D-5(E_1+\ldots+E_4)-2(E_5+\ldots+E_9), \label{I1234 0}\\
I_\beta(E_j) & \;= -E_j+\beta+\delta && \;=5D-2(E_1+\ldots+E_4)-E_j-(E_5+\ldots+E_9), \nonumber\\
&&& \hspace{8cm}\;\;j=1,\ldots,4,\label{I1234 1..4}\\
I_\beta(E_j) & \; = -E_j+\beta && \;=2D-(E_1+\ldots+E_4)-E_j,\;\;j=5,\ldots,9.\label{I1234 5..9}
\end{alignat}

\subsubsection{\boldmath$\beta=3D-2E_1-E_2-E_3-E_4-E_5-E_6$.} 

The space $V_1$ contains  $E_j$ for $j=7,8,9$, as well as $\gamma_1$, $\gamma_2$ from decompositions of the kind \eqref{beta split}:
\begin{equation}\label{Icubic split}
\gamma_1=D-E_1-E_i, \quad \gamma_2=2D-E_1-E_j-E_k-E_\ell-E_m,  
\end{equation}
where $i=2,\ldots,6$ and $\{j,k,\ell,m\}=\{2,3,4,5,6\}\setminus \{i\}$. This gives us eight linearly independent elements $E_j$ and $\gamma_1$ and proves  Theorem \ref{th I beta on Pic} in the present case. Vector space $V_0$ contains the following roots: 
\begin{eqnarray}
\lambda & = & E_i-E_j, \quad i,j\in\{2,3,4,5,6\}, \label{Icubic V0 1}\\
\lambda & = & D-E_i-E_j-E_k, \quad i,j,k\in\{2,3,4,5,6\}, \label{Icubic V0 2}\\
\lambda & = & D-E_1-E_i-E_j, \quad i\in\{2,3,4,5,6\}, \; j\in\{7,8,9\}. \label{Icubic V0 3}
\end{eqnarray}
Eight linearly independent ones are given, e.g., by \eqref{Icubic V0 1} with $i=2$ and $j=3,\ldots,6$,  \eqref{Icubic V0 2} with 
$\{i,j,k\}=\{2,3,4\}$, and  \eqref{Icubic V0 3} with $i=2$ and $j=7,8,9$.

We compute $I_\beta$ in the standard basis of ${\rm Pic}(X)$:
\begin{alignat}{2}
I_\beta(D) & \;= -D+ 3\beta+3\delta &&\;=17D-9E_1-6(E_2+\ldots+E_6)-3(E_7+E_8+E_9), \label{Icubic 0}\\
I_\beta(E_1) & \;= -E_1+ \beta+2\delta &&\;=9D-5E_1-3(E_2+\ldots+E_6)-2(E_7+E_8+E_9), \label{Icubic 1}\\
I_\beta(E_j) & \;= -E_j+\beta+\delta && \;=6D-3E_1-2(E_2+\ldots+E_6)-E_j-(E_7+E_8+E_9), \nonumber\\
&&& \hspace{8cm}j=2,\ldots,6,\label{Icubic 2..6}\\
I_\beta(E_j) & \; = -E_j+\beta && \;=3D-2E_1-(E_2+\ldots+E_6)-E_j,\;\;j=7,8,9.\label{Icubic 7..9}
\end{alignat}

\subsection{Scenario 2}

\subsubsection{\boldmath$\beta=H_1$ or \boldmath $\beta=H_2$}

If $\beta=H_1$, then we have, along with $I_\beta(\beta)=\beta$ and $I_\beta(\delta)=\delta$, eight relations $I_\beta(E_j)=\beta-E_j$. All ten relations are instances of \eqref{I beta on Pic}, which suffices to prove \eqref{I beta on Pic}. We compute $I_\beta$ in the standard basis of ${\rm Pic}(X)$: 
\begin{alignat}{2}
I_\beta(H_1) & \;=-H_1+2\beta &&\;=H_1,\\
I_\beta(H_2) & \; =-H_2+2\beta+\delta && \; =4H_1+H_2-(E_1+\ldots+E_8), \\
I_\beta(E_j) & \; =-E_j+\beta && \;=H_1-E_j, \;\; j=1,\ldots,8.
\end{alignat}
Similarly, if $\beta=H_2$, we have:
\begin{alignat}{2}
I_\beta(H_1) & \;=-H_1+2\beta+\delta &&\;=H_1+4H_2-(E_1+\ldots+E_8),\\
I_\beta(H_2) & \; =-H_2+2\beta && \; =H_2, \\
I_\beta(E_j) & \; =-E_j+\beta && \;=H_2-E_j, \;\; j=1,\ldots,8.
\end{alignat}

\subsubsection{\boldmath$\beta=H_1+H_2-E_1-E_2$}

The space $V_1$ contains  $E_j$ for $j=3,\ldots,8$ (six in total), as well as $\gamma_1$, $\gamma_2$ from decompositions of the kind \eqref{beta split}:
\begin{equation}
\gamma_1=H_1-E_i, \quad \gamma_2= H_2-E_j, \quad \{i,j\}=\{1,2\}.
\end{equation}
Any two of the latter can be taken as the 7-th and the 8-th linearly independent elements. This proves Theorem \ref{th I beta on Pic} for the present case. We compute $I_\beta$ in the standard basis of ${\rm Pic}(X)$:
\begin{alignat}{2}
I_\beta(H_1) & \;= -H_1+2\beta+\delta && \;=3H_1+4H_2-3(E_1+E_2)-(E_3+\ldots+E_8), \label{I12 2D H1}\\
I_\beta(H_2) & \;= -H_2+2\beta+\delta && \; =4H_1+3H_2-3(E_1+E_2)-(E_3+\ldots+E_8), \label{I12 2D H2}\\
I_\beta(E_1) & \;= -E_1+\beta+\delta && \; =3(H_1+H_2)-3E_1-2E_2-(E_3+\ldots+E_8), \label{I12 2D E1}\\
I_\beta(E_2) & \;= -E_2+\beta+\delta && \;=3(H_1+H_2)-2E_1-3E_2-(E_3+\ldots+E_8), \label{I12 2D E2}\\
I_\beta(E_j) & \;= -E_j+\beta && \;= H_1+H_2-E_1-E_2-E_j, \;\; j=3,\ldots,8. \label{I12 2D Ej}
\end{alignat}

\subsubsection{\boldmath$\beta=2H_1+H_2-E_1-E_2-E_3-E_4\;$ or \boldmath$\;\beta=H_1+2H_2-E_1-E_2-E_3-E_4$}

For definiteness, we consider the case $\beta=2H_1+H_2-E_1-E_2-E_3-E_4$ (the second one is obtained by interchanging the roles of $H_1$ and $H_2$). The space $V_1$ contains  $E_j$, $\beta-E_j$ for $j=5,6,7,8$, as well as $\gamma_1$, $\gamma_2$ from decompositions of the kind \eqref{beta split}:
\begin{equation}
\gamma_1=H_1-E_\ell, \quad \gamma_2= H_1+H_2-E_i-E_j-E_k,\quad \{i,j,k,\ell\}=\{1,2,3,4\}.
\end{equation}
Eight elements $E_j$, $j=5,6,7,8$ and $H_1-E_\ell$, $\ell=1,2,3,4$ are linearly independent. This proves Theorem \ref{th I beta on Pic} for the present case. We compute $I_\beta$ in the standard basis of ${\rm Pic}(X)$:
\begin{alignat}{2}
I_\beta(H_1) & \;= -H_1+2\beta+\delta && \;=5H_1+4H_2-3(E_1+\ldots+E_4)-(E_5+\ldots+E_8), \label{I2-1 H1}\\
I_\beta(H_2) & \;= -H_2+2\beta+2\delta && \; =8H_1+5H_2-4(E_1+\ldots+E_4)-2(E_5+\ldots+E_8), \label{I2-1 H2}\\
I_\beta(E_j) & \;= -E_j+\beta+\delta && \; =4H_1+3H_2-2(E_1+\ldots+E_4)-E_j-(E_5+\ldots+E_8), \nonumber\\
&&& \hspace{8cm} j=1,\ldots,4, \label{I2-1 E1}\\
I_\beta(E_j) & \;= -E_j+\beta && \;= 2H_1+H_2-(E_1+\ldots+E_4)-E_j, \;\; j=5,\ldots,8. \label{I2-1 E5}
\end{alignat}

\subsubsection{\boldmath$\beta=2(H_1+H_2)-2E_1-E_2-E_3-E_4-E_5$}

The space $V_1$ contains  $E_j$, $\beta-E_j$ for $j=6,7,8$, as well as $\gamma_1$, $\gamma_2$ from decompositions of the kind \eqref{beta split}:
\begin{eqnarray}
&&\gamma_1=H_1-E_1, \quad \gamma_2= H_1+2H_2-E_1-E_2-E_3-E_4-E_5,  \label{I2-2 gamma1 1}\\
&&\gamma_1=H_2-E_1, \quad \gamma_2= 2H_1+H_2-E_1-E_2-E_3-E_4-E_5,  \label{I2-2 gamma1 2}\\
&&\gamma_1=H_1+H_2-E_1-E_i-E_j, \quad \gamma_2= H_1+H_2-E_1-E_k-E_\ell, \nonumber\\
&& \hspace{8.5cm} \{i,j,k,\ell\}=\{2,3,4,5\}.  \label{I2-2 gamma1 3}
\end{eqnarray}
Elements $E_j$ and $\gamma_1$ from \eqref{I2-2 gamma1 1}, \eqref{I2-2 gamma1 2} and \eqref{I2-2 gamma1 3} with $(i,j)=(2,3), (2,4), (2,5)$ are linearly independent. This proves Theorem \ref{th I beta on Pic} for the present case. 
The space $V_0$ of differences of elements of $V_1$ contains, in particular, the following roots: 
\begin{eqnarray}
\lambda & = & H_1-H_2, \label{I2 2D H1-H2}\\
\lambda & = & E_i-E_j, \quad i,j\in\{6,7,8\}\;\;{\rm or}\;\; i,j\in\{2,3,4,5\}, \label{I2 2D Ei-Ej}\\
\lambda & = & H_1-E_1-E_j, \quad j\in\{6,7,8\}, \label{I2 2D H1-E1-Ej}\\
\lambda & = & H_1-E_i-E_j, \quad i,j\in\{2,3,4,5\}. \label{I2 2D H1-Ei-Ej}
\end{eqnarray}
Eight linearly independent ones among them, are given by, e.g., \eqref{I2 2D H1-H2}, \eqref{I2 2D Ei-Ej} with $i=6$ and $j=7,8$, as well as  with $i=2$ and $j=3,4,5$, \eqref{I2 2D H1-E1-Ej} with $j=6$, and \eqref{I2 2D H1-Ei-Ej} with $i=2$, $j=3$. 

We compute $I_\beta$ in the standard basis of ${\rm Pic}(X)$:
\begin{alignat}{2}
I_\beta(H_1) & \;= -H_1+2\beta+2\delta && \;=7H_1+8H_2-6E_1-4(E_2+\ldots+E_5)-2(E_6+E_7+E_8), \label{I2 2D H1}\\
I_\beta(H_2) & \;= -H_2+2\beta+2\delta && \; =8H_1+7H_2-6E_1-4(E_2+\ldots+E_5)-2(E_6+E_7+E_8), \label{I2 2D H2}\\
I_\beta(E_1) & \;= -E_1+\beta+2\delta && \; =6(H_1+H_2)-5E_1-3(E_2+\ldots+E_5)-2(E_6+E_7+E_8), \label{I2 2D E1}\\
I_\beta(E_j) & \;= -E_j+\beta+\delta && \;=4(H_1+H_2)-3E_1-2(E_2+\ldots+E_5)-E_j-(E_6+E_7+E_8), \nonumber\\
&&&
\qquad\qquad\qquad\qquad\qquad\qquad\qquad\qquad\qquad j=2,\ldots,5, \label{I2 2D E2345}\\
I_\beta(E_j) & \;= -E_j+\beta && \;= 2(H_1+H_2)-2E_1-(E_2+\ldots+E_5)-E_j, \;\; j=6,7,8. \label{I2 2D E678}
\end{alignat}

\subsubsection{\boldmath$\beta=3(H_1+H_2)-2E_1-2E_2-2E_3-2E_4-E_5-E_6$}

The space $V_1$ contains  $E_j$, $\beta-E_j$ for $j=7,8$, as well as $\gamma_1$, $\gamma_2$ from decompositions of the kind \eqref{beta split}:
$$
\gamma_1=H_1+2H_2-E_1-\ldots-E_4-E_j, \quad \gamma_2= 2H_1+H_2-E_1-\ldots-E_4-E_k,
$$
where $\{j,k\}=\{5,6\}$, and 
$$
\gamma_1=H_1+H_2-E_j-E_k-E_\ell, \quad \gamma_2= 2(H_1+H_2)-2E_i-E_j-E_k-E_\ell-E_5-E_6,
$$
where $\{i,j,k,\ell\}=\{1,2,3,4\}$. We have eight linearly independent elements $E_j$ and $\gamma_1$. This proves Theorem \ref{th I beta on Pic} for the present case. 
The space $V_0$ of differences of elements of $V_1$ contains, in particular, the following roots: 
\begin{eqnarray}
\lambda & = & H_1-H_2, \label{I3 2D H1-H2}\\
\lambda & = & E_5-E_6,  \;\; E_7-E_8, \label{I3 2D E5-E6}\\
\lambda & = & E_i-E_j, \quad i,j\in\{1,2,3,4\}, \label{I3 2D Ei-Ej}\\
\lambda & = & H_1-E_i-E_k, \quad i\in\{1,2,3,4\}, \;\;k\in\{5,6\}. \label{I3 2D H1-Ei-Ej}
\end{eqnarray}
One easily identifies eight linearly independent ones among them, say, \eqref{I3 2D H1-H2}, \eqref{I3 2D E5-E6}, \eqref{I3 2D Ei-Ej} with $i=1$ and $j=2,3,4$, and \eqref{I3 2D H1-Ei-Ej} with $i=1$, $j=5,6$. 

We compute $I_\beta$ in the standard basis of ${\rm Pic}(X)$:
\begin{alignat}{2}
I_\beta(H_1) & \;= -H_1+2\beta+3\delta && \;=11H_1+12H_2-7(E_1+\ldots+E_4)-5(E_5+E_6)-3(E_7+E_8), \label{I3 2D H1}\\
I_\beta(H_2) & \;= -H_2+2\beta+3\delta && \;=12H_1+11H_2-7(E_1+\ldots+E_4)-5(E_5+E_6)-3(E_7+E_8), \label{I3 2D H2}\\
I_\beta(E_j) & \;= -E_j+\beta+2\delta && \;=7(H_1+H_2)-4(E_1+\ldots+E_4)-E_j-3(E_5+E_6)-2(E_7+E_8), \nonumber\\ 
&&& \hspace{8.2cm} j=1,\ldots,4, \label{I3 2D E1234}\\
I_\beta(E_j) & \;= -E_j+\beta+\delta && \;= 5(H_1+H_2)-3(E_1+\ldots+E_4)-2(E_5+E_6)-E_j-(E_7+E_8),  \nonumber\\
&&& \hspace{9cm} j=5,6, \label{I3 2D E56}\\
I_\beta(E_j) & \;= -E_j+\beta && \;=3(H_1+H_2)-2(E_1+\ldots+E_4)-(E_5+E_6)-E_j,\;\; j=7,8. \label{I3 2D E78}
\end{alignat}

\subsection{Scenario 3}

\subsubsection{\boldmath$\beta=\Pi-E_1-E_2$}

We apply Proposition \ref{prop 2D to 3D} to transfer formulas \eqref{I12 2D H1}--\eqref{I12 2D Ej} for $\widetilde\beta=H_1+H_2-E_1-E_2$ in scenario 2 to the present context and obtain the following formulas for the action of $I_\beta$ on ${\rm Pic}(X)$:
\begin{alignat}{2}
I_\beta(\Pi) & \;= -\Pi+4\beta+2\delta && \;=7\Pi-6(E_1+E_2)-2(E_3+\ldots+E_8), \label{I12 3D Pi}\\
I_\beta(E_1) & \;= -E_1+\beta+\delta && \; =3\Pi-3E_1-2E_2-(E_3+\ldots+E_8), \label{I12 3D E1}\\
I_\beta(E_2) & \;= -E_2+\beta+\delta && \;=3\Pi-2E_1-3E_2-(E_3+\ldots+E_8), \label{I12 3D E2}\\
I_\beta(E_j) & \;= -E_j+\beta && \;= \Pi-E_1-E_2-E_j, \;\; j=3,\ldots,8. \label{I12 3D Ej}
\end{alignat}

\subsubsection{\boldmath$\beta=2\Pi-2E_1-E_2-E_3-E_4-E_5$}

Applying Proposition \ref{prop 2D to 3D}  to \eqref{I2 2D H1}--\eqref{I2 2D E678}, we find:
\begin{alignat}{2}
I_\beta(\Pi) & \;= -\Pi+4\beta+4\delta && \;=15\Pi-12E_1-8(E_2+E_3+E_4+E_5)-4(E_6+E_7+E_8), \label{I2 3D Pi}\\
I_\beta(E_1) & \;= -E_1+\beta+2\delta && \; =6\Pi-5E_1-3(E_2+E_3+E_4+E_5)-2(E_6+E_7+E_8), \label{I2 3D E1}\\
I_\beta(E_j) & \;= -E_j+\beta+\delta && \;=4\Pi-3E_1-2(E_2+E_3+E_4+E_5)-E_j-(E_6+E_7+E_8), \nonumber\\
&&&
\qquad\qquad\qquad\qquad\qquad\qquad\qquad\qquad\qquad j=2,\ldots,5, \label{I2 3D E2345}\\
I_\beta(E_j) & \;= -E_j+\beta && \;= 2\Pi-2E_1-(E_2+E_3+E_4+E_5)-E_j, \;\; j=6,7,8. \label{I2 3D E678}
\end{alignat}

\subsubsection{\boldmath$\beta=3\Pi-2E_1-2E_2-2E_3-2E_4-E_5-E_6$}

%
Applying Proposition \ref{prop 2D to 3D}  to \eqref{I3 2D H1}--\eqref{I3 2D E78}, we find:
\begin{alignat}{2}
I_\beta(\Pi) & \;= -\Pi+4\beta+6\delta && \;=23\Pi-14(E_1+\ldots+E_4)-10(E_5+E_6)-6(E_7+E_8), \label{I3 3D Pi}\\
I_\beta(E_j) & \;= -E_j+\beta+2\delta && \;=7\Pi-4(E_1+\ldots+E_4)-E_j-3(E_5+E_6)-2(E_7+E_8), \nonumber\\
&&& \hspace{8.1cm} \;\; j=1,\ldots,4, \label{I3 3D E1234}\\
I_\beta(E_j) & \;= -E_j+\beta+\delta && \;= 5\Pi-3(E_1+\ldots+E_4)-2(E_5+E_6)-E_j-(E_7+E_8), \nonumber\\
&&& \hspace{8.1cm}  \;\; j=5,6, \label{I3 3D E56}\\
I_\beta(E_j) & \;= -E_j+\beta && \;=3\Pi-2(E_1+\ldots+E_4)-(E_5+E_6)-E_j,\;\; j=7,8. \label{I3 3D E78}
\end{alignat}

\end{appendix}

\end{document}